\documentclass{sig-alternate}

\usepackage{ifthen}
\usepackage{pst-all}
\usepackage{amsmath}
\usepackage{theorem}
\usepackage{subfigure}
\usepackage{epsfig}
\usepackage{amssymb}
\usepackage{verbatim}
\usepackage{appendix}
\numberwithin{equation}{section} %% Comment out for sequentially-numbered
\numberwithin{figure}{section} %% Comment out for ly-numbered
\newtheorem{thm}{Theorem}
  \newtheorem{lem}[thm]{Lemma}
  
  \newtheorem{example}[thm]{Example}
  
  \newtheorem{rem}[thm]{Remark}
  
  \newtheorem{fact}[thm]{Fact}

\begin{document}

\clubpenalty=10000

\widowpenalty = 10000 

\title{Robust Mechanisms for Risk-Averse Sellers}

%\subtitle{[Extended Abstract]
%\titlenote{A full version of this paper is available as
%\textit{Author's Guide to Preparing ACM SIG Proceedings Using
%\LaTeX$2_\epsilon$\ and BibTeX} at
%\texttt{www.acm.org/eaddress.htm}}}

\numberofauthors{2}
\author{
\alignauthor Mukund Sundararajan\\
\affaddr{ Google Inc., \\ Mountain View, CA, USA }
\email{mukunds@google.com}.
\alignauthor Qiqi Yan\thanks{Supported by a
Stanford Graduate Fellowship}\\
\affaddr{Department of Computer Science, \\
Stanford University, \\
Stanford, CA, USA}
\email{qiqiyan@cs.stanford.edu}.
}

\maketitle

\begin{abstract}
The existing literature on optimal auctions focuses on optimizing the \emph{expected revenue} of the seller, and is appropriate for risk-neutral sellers. In this paper, we identify good mechanisms for \emph{risk-averse} sellers.  As is standard in the economics literature, we model the risk-aversion of a seller by endowing the seller with a monotone concave utility function. We then seek robust mechanisms that are approximately optimal for all sellers, no matter what their levels of risk-aversion are.

We have two main results for multi-unit auctions with unit-demand bidders whose valuations are drawn i.i.d.\ from a regular distribution. First, we identify a posted-price mechanism called the Hedge mechanism, which gives a universal constant factor approximation; we also show for the unlimited supply case that this mechanism is in a sense the best possible. Second, we show that the VCG mechanism gives a universal constant factor approximation when the number of bidders is even only a small multiple of the number of items. Along the way we point out that Myerson's characterization of the optimal mechanisms fails to extend to utility-maximization for risk-averse sellers, and establish interesting properties of regular distributions and monotone hazard rate distributions.

\end{abstract}

\category{J.4}{Computer Applications}{Social and Behavioral Sciences}[Economics]

\terms{Economics, Theory, Algorithms}

\keywords{risk-aversion, optimal auctions, revenue maximization, utility}

\section{Introduction\label{sec:intro}}

Auction theory (cf.~\cite{M81, BK96}) typically seeks to optimize the seller's \emph{expected} revenue, which presumes that the seller is \emph{risk-neutral}. The focus of this work is to identify good auction mechanisms for sellers who care about the riskiness of the revenue in addition to the magnitude of the revenue\footnote{We seek ex-post incentive compatible mechanisms. This is in contrast to the standard Bayesian auction theory literature (cf.~\cite{M81,BK96}) that studies Bayesian incentive compatible mechanisms. In our auctions, bidders will therefore maximize utility by truth-telling, and do not have to deal with uncertainty or risk; our model of risk applies only to sellers.}. 

There is an inherent trade-off between the magnitude and riskiness of revenue. Consider the auction of a single-item to a bidder whose valuation is drawn from the uniform distribution over the interval $[0,1]$. Recall that every truthful single-bidder mechanism offers the bidder a take-it-or-leave-it price. If the seller is \emph{risk-neutral} and cares about mean revenue, we must select a price $p$ that maximizes the product of the price $p$ times the probability of sale $1-p$. The price $p=1/2$ is optimal here, achieving a mean revenue of $1/4$, but it yields zero revenue with probability $1/2$. Prices lower than $1/2$ reduce the expected revenue, but increase the certainty with which positive revenue is obtained.

A systematic and standard (cf. Stiglitz and Rothschild~\cite{Stig1}) way to express a bidder's trade-off between the magnitude and riskiness of revenue is to endow the seller with a concave utility function $u:[0,\infty)\to[0,\infty)$ and seek to maximize the seller's expected utility. We will assume throughout that this utility function is monotone and normalized in the sense that $u(0)=0$. Let $Rev(M,\mathbf{v})$ denote the revenue of mechanism $M$ for the  bid-profile $\mathbf{v}$, then the expected utility of $M$ w.r.t.\ a utility function $u$ is $E_{\mathbf{v}}[u(Rev(M,\mathbf{v}))]$. The concavity of the utility function models \emph{risk-aversion}. For instance, the optimal single-bidder mechanism for the utility function $u(x)=\sqrt{x}$ sets a price $p=1/3$ and maximizes the expected utility $ \sqrt{p} \cdot (1-p)$. Increasing the concavity of the utility function increases the emphasis on risk-aversion---the optimal price for the cube-root utility function is $p=1/4$. The linear utility function $u(x)=x$ models a \emph{risk-neutral} seller.
\emph{The goal of this paper is to identify truthful mechanisms that are simultaneously good for the class of all risk-averse agents, i.e., we look for mechanisms that yield near-optimal expected utility for all possible concave utility functions.} 

A useful byproduct of such a guarantee is that we do not need to know the seller's utility function in order to deploy the mechanism. This is useful when the auctioneer is conducting the auction on behalf of a seller (as in the case of eBay), when the seller does not know its utility function precisely, or when the seller's risk attitude changes with time.

The following example illustrates the challenge in the context of a single-item single-bidder auction. Consider two sellers with utility functions $u_{\text{risk-neutral}}(x)=x$, which expresses risk-neutrality, and $u_{\text{risk-averse}}(x)=\min(x,\epsilon)$ for some very small $\epsilon>0$, which expresses strong risk-aversion. Suppose, as before, that there is a single bidder whose valuation is drawn from the uniform distribution with support $[0,1]$. The \emph{unique} optimal mechanism for the first utility function makes a take-it-or-leave-it offer of $1/2$. This gives the first seller a utility of $1/4$, and gives the second seller a utility of $\epsilon \cdot (1-F(1/2)) = \epsilon/2$. Lowering the price to $\epsilon$ improves the second seller's utility to $(1 - \epsilon) \cdot \epsilon$, but reduces the first seller's utility from $1/4$ to $(1 - \epsilon) \cdot \epsilon$. Our challenge in general is to identify mechanisms that simultaneously appease sellers with different levels risk-aversion, ranging from risk-neutral sellers who care about expected revenue to very risk-averse ones who only care about the certainty with which a positive revenue is obtained.

\subsection{Organization}
Section~\ref{sec:prelim} describes our auction model, our distributional assumptions and formalizes our auction objective.
Section~\ref{sec:benchmark} describes the difficulty in characterizing our benchmark and defines a stronger, simpler benchmark. Section~\ref{sec:spm} identifies universally approximate posted-price mechanisms for unlimited and limited supply. Section~\ref{sec:vcg} bounds the universal approximation of the VCG mechanisms for multi-unit auctions as a function of the ratio of the number of bidders to the number of items. Section~\ref{sec:open} concludes with open directions.

\section{Preliminaries}\label{sec:prelim}
\subsection{Auction Model}
Our investigation focuses on multi-unit auctions. We adopt the following standard auction model. There are $n$ unit-demand bidders $1,2,\dots,n$, and $k$
identical indivisible items for sale. A bidder $i$ has a private
valuation $v_{i}$ for winning an item, and $0$ for losing. A mechanism
$\mathcal{M}=(\mathbf{x},\mathbf{p})$ first collects a bid $b_{i}$
from each bidder $i$, then determines the winners by the allocation
rule $\mathbf{x}:\mathbf{b}\to\{0,1\}^{n}$, i.e., bidder $i$ wins
an item if and only if $x_{i}(\mathbf{b})=1$, and finally uses the
payment rule $\mathbf{p}:\mathbf{b}\to[0,\infty)^{n}$ to charge each
bidder $i$ a price $p_{i}(\mathbf{b})$. We will focus our attention on ex post incentive compatible, a.k.a., truthful,\footnote{For any possible bids $\mathbf{b}_{-i}$ of the other bidders, bidder $i$ always maximizes her utility $v_{i}\cdot x_{i}(\mathbf{b})-p_{i}(b_i;\mathbf{b}_{-i})$, by setting her bid $b_i$ to be her true valuation $v_i$.} and ex post individual-rational\footnote{A bidder is never
charged more than her bid, and is only charged when she wins.} mechanisms. Hence we will use the terms bid and valuation interchangeably. We make the standard assumption that valuations are drawn i.i.d.\ from a distribution $F$. The distribution $F$ is known to the seller, but the valuations can be known only to buyers.

\subsection{Auction Objective}

Let $Rev(M,\mathbf{v})$ denote the revenue of mechanism $M$ for the input bid-profile $\mathbf{v}$. Then the expected revenue of $M$ is $E_{\mathbf{v}}[u(Rev(M,\mathbf{v}))]$. Notice that the expectation is over the bids (or valuations), which is the standard auction objective in Bayesian revenue maximization. We model the risk-attitude of a specific seller by endowing the seller with a concave utility function $u:[0,\infty)\to[0,\infty)$. We will assume throughout that this utility function is monotone and normalized in the sense that $u(0)=0$. Then the expected utility of $M$ w.r.t.\ a utility function $u$ is $E_{\mathbf{v}}[u(Rev(M,\mathbf{v}))]$. As discussed in the introduction, the concavity of the utility function models risk-aversion.

Recall that the goal of this paper is to identify truthful mechanisms that are simultaneously good for the class of all risk-averse agents, i.e., we look for mechanisms that yield near-optimal expected utility for all possible  concave normalized utility functions. More precisely, for each risk-averse seller, the truthful mechanism $M_u^*$ that maximizes the seller's expected utility is a benchmark against which we measure our proposed mechanism (say $M$)--we quantify the goodness of this mechanism for this seller by the approximation ratio $U(M)/U(M^*_u)$, where $U(X)$ denotes the expected utility of mechanism $X$. The goodness of the mechanism is then the worst-case approximation ratio over all concave utility functions, i.e. $\rho=min_{u} U(M)/U(M^*_u)$; in this case, we will say that the mechanism is a universal $\rho$-approximation. For each of the auction settings we consider, we will try to find a mechanism $M$ that maximizes $\rho$.

\subsection{Distributional Assumptions}
\label{sec:dist}

For technical convenience, we will assume that the distribution $F$ has a smooth positive density function, and has non-negative support.
We will in addition assume that the distribution $F$ from which the valuation is drawn satisfies a standard \emph{regularity} condition  (cf.~\cite{M81},~\cite{BK96}). 

Every distribution function $F$ corresponds to a revenue function $R$ from domain $[0,1]$ (or $(0,1]$ if the support of $F$ is infinite) to the non-negative reals defined as follows: for all $q$, $R_F(q) =  q\cdot F^{-1}(1-q)$. (we will drop the subscript when it is clear from the context) Note that $R(0)=0$ (or $R(q) \rightarrow 0$ as $q \rightarrow 0$) and $R(1)=0$, and we can often define a distribution $F$ by specifying the corresponding $R_F(\cdot)$ function. We say a distribution $F$ is \emph{regular} if the revenue function $R_F(\cdot)$ w.r.t.\ $F$ is strictly concave. This is also equivalent to the more commonly used definition that virtual valuation $\phi_{F}(v)=v-1/h(v)$ is nondecreasing in $v$, where $h(v)=\frac{f(v)}{1-F(v)}$ is the hazard rate function w.r.t.\ $F$. We say $F$ satisfies the \emph{monotone hazard rate} condition (or simply $F$ is m.h.r.), if $h(v)$ is nondecreasing in $v$. Many important distributions are regular and m.h.r, including uniform, exponential, normal, while other distributions such as some power-law distributions are regular but not m.h.r.~\cite{Ewe09}.

To justify our use of the regularity assumption, the following example shows that no universal constant factor approximation is possible without assumptions on the distribution $F$.

\begin{example}
Recall the utility functions $u_{\text{risk-neutral}}$ and $u_{\text{risk-averse}}$ defined in the introduction. Define $R$ as $R(0)=R(1)=0$, $R(\epsilon)=1$, $R(2\epsilon)=\epsilon$, $R(1-\epsilon)=\epsilon$, and let $R$ be linear in all four intervals between these five points; here '$\epsilon$' refers to the quantity in the definition of $u_{\text{risk-averse}}$ (see introduction). Smoothen $R$ by a negligible amount such that the corresponding $F$ function satisfies our smoothness assumption on distributions. Consider a single bidder whose valuation function is drawn from $F$, which is clearly an irregular distribution.

Thus to achieve a constant fraction of optimal utility for $u_{\text{risk-neutral}}$ means that we have to sell with a probability in the range of $[0,\epsilon]$, i.e., at a price of at least $1/2$, which implies that we get at most $2\epsilon^2$ utility for $u_{\text{risk-averse}}$.
\end{example}

\subsection{Results and Techniques} \label{sec:results}

We first show that the 'virtual value' based approach employed by Myerson~\cite{M81} for the risk-neutral case extends to risk-averse single-item auctions, but not (to the best of our knowledge) to auctions of two or more items (see Section~\ref{sec:benchmark}). We then present three results. First, when the supply is unlimited (or equivalently, the number of items $k$ is equal to the number of bidders $n$), we identify a mechanism called the Hedge mechanism that is a universal $1/2$-approximation (see Theorem~\ref{thm:unlimited}). The ratio improves to nearly $0.7$ with the assumption that the distribution satisfies a standard hazard rate condition. The Hedge mechanism is a posted-price mechanism, which offers every bidder a take-it-or-leave-it offer $p$ in a sequential order so long as supply lasts. We choose the price $p$ to be less than the optimal price for a risk-neutral seller so as to guarantee a good probability of sale to any bidder at a good revenue level. Moreover, this mechanism is the best possible in the sense that no mechanism can be a universal $\rho$-approximation for $\rho>1/2$ (see Theorem~\ref{thm:lb}). This impossibility result identifies a certain heavy-tailed regular distribution, called the left-triangle distribution that exhibits the worst-case  trade-off between riskiness and magnitude of revenue over all regular distributions. Second, when the supply is limited (number of items $k$ is less than the number of bidders $n$), we identify a sequential posted-price mechanism that gives a universal $1/8$-approximation by modifying the Hedge mechanism to handle the supply constraint (see Theorem~\ref{thm:limited-regular-1}). The key to this modification is to use a certain limited supply auction to guide the choice of the posted price. Third, we will show that the VCG mechanism \cite{vic-61,cla-71,gro-73} yields a universal approximation ratio close to $1/4$ under moderate competition, i.e., when $n$ is a reasonable multiple of $k$ (see Theorem~\ref{thm:vcg}). Recall that for a $k$-item auction the VCG mechanism is a $k+1$-st price auction, in which the top $k$ bidders win and get charged the $k+1$-st highest bid. We prove our result by establishing a probability bound for the $k+1$-st order statistic of $n$ i.i.d.\ draws from a regular distribution.

\subsection{Related Work}

Myerson~\cite{M81} identifies the optimal single-item mechanism for a risk-neutral seller and has inspired a large body of work (cf. Chapter 13 from~\cite{AGTBook}). 

There is some work that deals with risk in the context of auctions. Eso~\cite{risk99} identifies an optimal mechanism for a risk-averse seller, which always provides the same revenue \emph{at every} bid vector by modifying Myerson's optimal mechanism; unfortunately, this mechanism does not satisfy ex-post (or even ex-interim) individual rationality, and charges bidders even when they lose. Maskin and Riley~\cite{buyers81} identifies the optimal Bayesian-incentive compatible mechanism for a risk-neutral seller when the \emph{bidders} are risk-averse. In our model, we identify mechanisms that are ex-post incentive compatible. So the buyers optimize their utility bidding truthfully for every realization of the valuations, and thus have no uncertainty or risk to deal with. Hu et al.\ \cite{hu10} studies risk-aversion in single-item auctions. Specifically, they show for both the first and second price mechanisms that the optimal reserve price reduces as the level of risk-aversion of the seller increases. In contrast, we identify the optimal truthful mechanism for a risk-averse seller in a single-item auction in Section~\ref{sec:benchmark} (it happens to be a second price mechanism with a reserve), study auctions of two or more items and identify mechanisms that are simultaneously approximate for all risk-averse sellers.

An alternative simpler model of risk different from the one we adopt is to optimize for a trade-off between the mean and the variance of the auction revenue, i.e., $E[R] - t \cdot Var[R]$. However, as Section 2A in Stiglitz and Rothschild~\cite{Stig2} shows, this approach does not capture all the types of behavior intuitively consistent with risk-aversion, because this approach restricts the form of seller utility functions. Our model of risk-aversion is inspired in part by Stiglitz and Rothschild~\cite{Stig1}.

There is significant literature on prior-free optimal auctions (see Chapter 13 from~\cite{AGTBook}). In this framework, the benchmark (in the unlimited supply case,  the revenue from the optimal price for that bid vector constrained to serve at least $2$ bidders) is defined independently for each bid vector, and the performance of the mechanism is measured worst-case over all bid-vectors. In contrast, in our framework, as in all Bayesian auction theory, the mechanism's performance is measured in expectation over the distribution of the bids. However, we believe that it is worth investigating the risk properties of the mechanisms proposed in this literature, which ought to yield universal constant factor approximations in several auction settings.

Finally, we mention papers that inspire our proof techniques. Chawla et al.~\cite{CHMS10} proposes posted-price mechanisms, and it uses Myerson's mechanism to guide the selection of the prices. We use a similar idea in Section~\ref{sub:limited}. Bulow and Klemperer~\cite{BK96} shows that the VCG mechanism with $k$ extra bidders yields better expected revenue than the optimal mechanism so long as the bidder valuations are drawn i.i.d.\ from a regular distribution.  Dughmi et al.~\cite{Dughmi09} extends the result of Bulow and Klemperer~\cite{BK96} to matroid settings, and introduces the problem of designing markets with good revenue properties. We use ideas from these papers to bound the performance of the VCG mechanism in Section~\ref{sec:vcg}. The characterization of regular distributions in terms of concave revenue functions is implicit in Myerson~\cite{M81}, and is used explicitly in Chawla et al.~\cite{CHMS10} and Dhangwatnotai et al.~\cite{DRY10}.

\section{On Utility-Optimal Mechanisms}
\label{sec:benchmark}

Recall from the introduction that we would like to design mechanisms that yield a good approximation of the optimal expected utility for each concave utility function. Our \emph{benchmark} for a specific utility function $u$ is the truthful individually rational mechanism that maximizes the expected utility w.r.t.\ $u$. In this section we focus on getting a handle on such a mechanism for a \emph{fixed} utility function $u$. We show that the result of Myerson~\cite{M81} can be extended to identify the optimal mechanism for the single item case, but not for auctions of two or more items. For the rest of the paper, we use the stronger simpler benchmark from Fact~\ref{fact:upperbound}.

Myerson's characterization says that the \emph{expected revenue} of \emph{any} truthful mechanism equals the expected total virtual valuation served by the mechanism. It generates a prescription for the allocation and payments of the optimal \emph{risk-neutral} truthful mechanism on a specific input bid vector. In the single-item case, to generalize Myerson's characterization to auctions with risk-averse sellers, we generalize the notion of virtual valuation to take risk-aversion into account: given a distribution $F$ and a concave utility function $u$, we define the \emph{virtual utility} function as $\phi_{F}^{u}(v)=u(v)-u'(v)/h(v)$. As in the case of virtual valuations, the virtual utility $\phi_{F}^{u}(v)$ is
the derivative $\frac{d}{d\left(1-F(v)\right)} u(v) \left(1-F(v)\right)$ of the expected utility from a bid-independent take-it-or-leave-it offer $v$ to a single bidder. We then have the following:

\begin{lem}
\label{lem:virtual_utility}In a single-item auction, for any mechanism
$M=(\mathbf{x},\mathbf{p})$ and concave utility function $u$, the expected utility of the mechanism, $E_{\mathbf{v}}[u(Rev(M,\mathbf{v}))]$, is equal to the expected virtual valuation served $E_{\mathbf{v}}[\sum_{i}\phi_{F}^{u}(v_{i})\cdot x_{i}(\mathbf{v})]$.\end{lem}
\begin{proof}
The expected utility of the mechanism is: \begin{eqnarray*}
E_{\mathbf{v}}[u(Rev(M,\mathbf{v}))] & = & E_{\mathbf{v}}[u(\sum_{i}p_{i}(\mathbf{v}))]\\
 & = & \sum_{i}E_{\mathbf{v}_{-i}}[E_{v_{i}}[u(p_{i}(\mathbf{v}))]]\\
 & = & \sum_{i}E_{\mathbf{v}_{-i}}[E_{v_{i}}[\phi_{F}^{u}(v_{i})\cdot x_{i}(\mathbf{v})]]\\
 & = & E_{\mathbf{v}}[\sum_{i}\phi_{F}^{u}(v_{i})\cdot x_{i}(\mathbf{v})]\end{eqnarray*}

Here the second equality holds because we sell to at most 1 bidder. The third equality holds because when $\mathbf{v}_{-i}$ is fixed, the mechanism induces a fixed offer price, say $p'$, for bidder $i$. So $E_{v_{i}}[u(p_{i}(\mathbf{v}))] = u(p') \left(1-F(p') \right)$, which is equal to $\int_{p'}^\infty \left(u(v)-u'(v)/h(v)\right)f(v) dv$, which is  $E_{v_{i}}[\phi_{F}^{u}(v_{i})\cdot x_{i}(\mathbf{v})]$, the expected virtual utility we get from bidder $i$.
\end{proof}

We can now use the lemma to show that the optimal mechanism for a seller with utility function $u$ is a second price auction with a reserve price---a mechanism that is well-known to be truthful. Consider the second price mechanism with a reserve  $r_u^*$, where $r_u^*$ solves that $\phi_{F}^{u}(r_u^*)=0$. When the distribution is regular, the virtual utility function is nondecreasing in the valuation (see Lemma~\ref{lem:u-regular} in the appendix). So the above mechanism allocates the item to the bidder with the highest virtual utility, so long as there is at least one bidder with non-negative virtual utility. (When the distribution is not regular, and in particular when the virtual utility function is not monotone, one can apply the ironing procedure of Myerson to identify the optimal mechanism as the one that maximizes the total ironed virtual utility served.)

In Section~\ref{sec:vcg} we will present another application of the above characterization that shows that the single-item Vickrey auction has good revenue properties. However, this characterization does not extend to auctions where more than one items are for sale. The first step of the proof of Lemma~\ref{lem:virtual_utility}, which sums the contributions of the bidders independently, only works because a single-item auction sells to and charges at most one bidder. When there are more than one items for sale, that step is still sound if the utility function is linear (the risk-neutral case), but it does not work for strictly concave utility functions.

We now identify an upper bound on the expected utility of utility-optimal mechanism that applies to auction settings beyond single-item auctions. We will use this upper bound as a benchmark for analysis. For any mechanism $M$ and concave utility function $u$, the expected utility of the mechanism $E_{\mathbf{v}}[u(Rev(M,\mathbf{v}))]$ is upper-bounded by the utility function applied to the expected revenue $u(E_{\mathbf{v}}[Rev(M,\mathbf{v})])$ by Jensen's inequality, which is then upper-bounded by the utility function applied to the expected revenue of Myerson's revenue-optimal mechanism $Mye$, $u(E_{\mathbf{v}}[Rev(Mye,\mathbf{v})])$, because a utility function is monotone. So we have the following:

\begin{fact}\label{fact:upperbound} For any mechanism $M$\footnote{We shall only work with deterministic mechanisms, but in fact we can allow the mechanism here to be randomized.},
 and any concave
utility function $u$, the expected utility of $M$ is upper-bounded
by the utility function applied to the expected revenue of Myerson's
mechanism, i.e., $E_{\mathbf{v}}[u(Rev(M,\mathbf{v}))]\leq u(E_{\mathbf{v}}[Rev(Mye,\mathbf{v})])$.\end{fact}

\section{Universally Approximate Sequential Posted-Price Mechanisms}
\label{sec:spm}

In this section we propose sequential posted-price mechanisms (or SPM in short) for multi-unit auctions. In an SPM, a take-it-or-leave-it price is offered to each bidder one by one in arbitrary order, as long as supply lasts. An obvious advantage of such mechanisms is that they can be applied to both offline and online settings and are collusion-resistant in the sense of Goldberg and Hartline~\cite{GH-05}.

\subsection{The Unlimited Supply Case\label{sub:unlimited}}

Fix a regular distribution $F$ from which the valuations are drawn i.i.d. We now identify an SPM that offers every bidder the same take-it-or-leave-it offer $p$, and show that this mechanism is universally $1/2$-approximate for all regular distributions, and $0.69$-approximate for all m.h.r.\ distributions. Let $p^*$ is the optimal price that maximizes $p(1-F(p))$, and $q^*=1-F(p)$. Setting the offer price $p$ to be $p^*$ yields the optimal expected revenue, but the probability of sale for each bidder can be very low. Intriguingly, we find that reducing the offer price to $p^*q^*$ is optimal, i.e., the discount factor is precisely the probability of sale at the optimal price for a risk-neutral seller in a single item-single bidder auction. We call this SPM with posted price $p=p^*$ the \emph{Hedge Mechanism}. Theorem~\ref{thm:unlimited} shows that this achieves a universal $1/2$ approximation for regular distributions (universal $0.69$-approximation for m.h.r.\ distributions), and Theorem~\ref{thm:lb} shows that we cannot do better.

\includegraphics{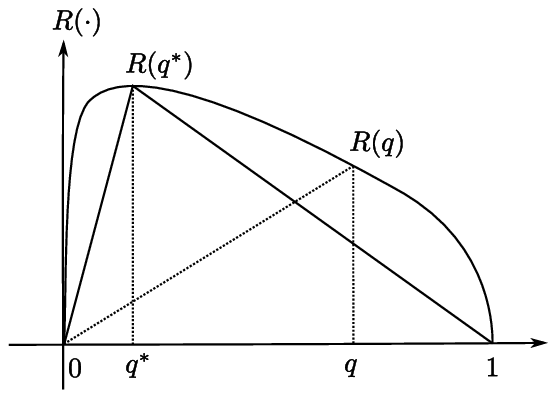}

To analyze the performance of the Hedge mechanism, the following property
of regular distributions is crucial.

\begin{lem}
\label{lem:regular}
For all regular distribution $F$, we have $1-F(p^{*}q^{*})\geq 1/2$.
\end{lem}
\begin{proof}
Let $q=1-F(p^* q^*)$. Note that $q\geq q^*$ because $p^* q^*\leq p^*$. Let $R(\cdot)$ be $F$'s revenue function, which is concave by regularity. The fact that $q\geq 1/2$ follows from the following inequalities:

\begin{eqnarray*}
q &=& R(q)/(p^*q^*)\\
&\geq& \left( R(q^*)\frac{1-q}{1-q^*}+ R(1)\frac{q-q^*}{1-q^*}  \right)/(p^*q^*)\\
&\geq& \left( (p^*q^*)\frac{1-q}{1-q^*}\right)/(p^*q^*)\\
&=&\frac{1-q}{1-q^*} \\
&\geq& 1-q
\end{eqnarray*}

The first step is by the definition of $q$. The second step is by the concavity of $R$. (In the above figure, note that $(q,R(q))$ is above the line segment connecting $(q^*,R(q^*))$ and $(1,R(1))$). The third step is because $R(q^*)=p^*q^*$ and $R(1)$ is non-negative.
\end{proof}

When the distribution $F$ is further assumed to be m.h.r., we can improve the constant to $e^{-1/e}$.

\begin{lem}
\label{lem:mhr} For any m.h.r.\ distribution $F$, let $p^*$ maximize $p (1-F(p))$ and $q^* = 1-F(p^*)$. Then we have that $1-F(p^{*}q^{*})\geq e^{-1/e}\approx 0.6922$.\end{lem}
\begin{proof}

W.l.o.g., we can let $p^{*}=1$ by scaling the valuation space.
Let cumulative hazard rate function $H(x)$ be $\int_{0}^{x}h(t)dt$, and note that the monotone hazard rate condition implies that $H(x)$ is monotone, convex, and normalized ($H(0)=0$). Note that at the price $p^{*}=1$, the virtual valuation is 0, i.e., $1-1/h(1)=0$. So $h(1)=1$. Further, the function $h$ is nondecreasing. So $H(1)=\int_{0}^{1}h(t)dt\leq1\cdot h(1)=1$.
Our claim follows from the following inequalities:

\begin{eqnarray*}
q &=& 1-F(p^* q^*)\\
&=& 1-F(q^*)\\
&=& e^{-H(q^{*})}\\
&=&e^{-H(1-F(p^{*}))}\\
&=&e^{-H(e^{-H(p^{*})})}\\
&=& e^{-H(e^{-H(1)})}\\
&\geq& e^{-H(1)e^{-H(1)} }\\
&\geq& e^{-1/e}
\end{eqnarray*}
The first step is by definition of $q$ and $R(q)$. The second and sixth steps are because $p^*=1$. The third and fifth steps are because the distribution function can be written in terms of the cumulative hazard rate function: $F(x) = 1 - e^{-H(x)}$. The seventh step is because $H(e^{-H(1)}) \leq e^{-H(1)}H(1)$ by the convexity of $H$ and that $H(1)\leq1$. The last step holds because $e^{-x}\cdot x$ is at most $1/e$ for $x\in[0,1]$.
\end{proof}

We now use the bounds in the previous two lemmas to complete the proof of the theorem.

\begin{thm}
\label{thm:unlimited}In a multi-unit auction with unlimited
supply, where bidders' valuations are drawn i.i.d.\ from a regular (or m.h.r)
distribution $F$, the $Hedge$ mechanism is a universal $0.5$ (or $e^{-1/e} \approx0.6922$)-approximation. \end{thm}
\begin{proof}

We prove for the regular case; for the proof of the m.h.r.\ case we simply use the bound from Lemma~\ref{lem:mhr} instead of the bound from Lemma~\ref{lem:regular}. Fix a concave utility function $u$. For each bidder $i$, let 0-1 random variable $X_{i}$ indicate whether
bidder $i$'s bid is at least $p^*q^*$.

\begin{eqnarray*}
\mbox{Expected Utility of $Hedge$} &=&
E[u(\sum_{i}X_{i}\cdot p^{*}q^{*})] \\
&\geq& E[\frac{\sum_{i}X_{i}}{n}]\cdot u(np^{*}q^{*}) \\
&\geq& 0.5\cdot u(np^{*}q^{*}) \\
&\geq& 0.5\cdot \text{Optimal Expected Utility}
\end{eqnarray*}

The first step is because the sale price is $p^*q^*$. The second step is by monotonicity and concavity of $u$ and because $0\leq\sum_{i}X_{i}\cdot p^{*}q^{*}\leq np^{*}q^{*}$. The third step is by Lemma~\ref{lem:regular}, and hence $E[\sum_{i}X_{i}]\geq n/2$. Applying Fact~\ref{fact:upperbound} completes the proof.
\end{proof}

\begin{rem}
If bidders' valuations are drawn from non-identical but independent regular distributions, we can identify distinct offer prices for each bidder $i$, $p^*_i \cdot q^*_i$, (here $p^*_i$ is the price that maximizes the expected revenue in a single bidder-single item auction with bidder $i$; and $q^*_i$ is the sale probability at that price), such that the guarantee in Theorem~\ref{thm:unlimited} holds. 
\end{rem}

The following lemma shows that the ratios in Theorem~\ref{thm:unlimited} cannot be improved. The proof identifies a certain left-triangle distribution that exhibits worst-case behavior over regular distributions, and shows that the exponential distribution exhibits worst-case behavior over all m.h.r.\ distributions. The proof elucidates why the price $p^*q^*$ is critical for the single-bidder case and justifies its use in the Hedge mechanism.

\begin{thm} \label{thm:lb}
There exists a regular (or m.h.r) distribution such that no mechanism yields a universal approximation with ratio larger than than $1/2$ (or $e^{-1/e} \approx0.6922$) for a single-bidder single-item auction, respectively.
\end{thm}

\begin{proof}

Consider a single-item single-bidder auction. Consider two possible seller utility functions, $u_{\text{risk-neutral}}$ and $u_{\text{risk-averse}}$, as defined in the introduction. The optimal utility w.r.t.\ $u_{\text{risk-neutral}}$
is $p^{*}q^{*}$, achieved at price $p^{*}$, and the optimal utility
w.r.t.\ $u_{\text{risk-averse}}$ is roughly $\epsilon$ (as $\epsilon \rightarrow 0$), achieved at price $\epsilon$.

We argue that the sale probability $q = 1-F(p^{*}q^{*})$ at the price $p^{*}q^{*}$ is an upper-bound on the best universal approximation possible. The expected revenue at price $p^{*}q^{*}$ is $q p^{*}q^{*}$. So, the approximation ratio for the risk-neutral seller is precisely $q$. The expected utility for the risk-averse seller at price $p^{*}q^{*}$ is roughly $\epsilon q$. So, the approximation ratio for this seller is also $q$. 
Now suppose a price lower than $p^* q^*$ is offered. Then the expected revenue deteriorates, and the approximation ratio for the risk-neutral seller drops below $q$. On the other hand, suppose a price higher than $p^* q^*$ is offered. Then the sale probability drops below $q$, and so does the approximation ratio for the risk-averse seller.

Then it suffices to show that there is regular distribution with sale probability 1/2 at price $p^*q^*$, and there is an m.h.r.\ distribution with sale probability $e^{-1/e}$ at price $p^* q^*$.
First we define the \emph{left-triangle distribution} via its revenue function $R_L(\cdot)$ as follows. Let $R_L(0)=R_L(1)=0$, $R_L(\epsilon)=1$ for some small $\epsilon>0$, and let $R_L(q)$ be piecewise linear between these points, and smoothen it by a negligible amount to make sure that the corresponding $F$ is a valid distribution. (It is essentially a shifted Pareto distribution.)	So $p^*q^*$ is $1$, and clearly the sale probability at price $1$ is roughly $1/2$.

Second, consider the exponential distribution $F(p)=1-e^{-p}$, which satisfies the monotone hazard rate condition. Note that $p^*=1$ and $q^*=1/e$, and it follows that $1-F(p^*q^*)=e^{-1/e}$.
\end{proof}

\begin{rem}
Our bounds in Theorem~\ref{thm:lb} and Theorem~\ref{thm:unlimited} are worst-case over the number of bidders $n$, and the mechanism we propose does not require knowledge of $n$. In general, the knowledge of $n$ is useful: As $n$ increases it makes sense to increase the price from the heavily discounted price $p^*q^*$ towards the optimal risk-neutral price $p^*$, because for large $n$, the resulting revenue as a random variable is well concentrated.
\end{rem}

\subsection{\label{sub:limited}The Limited Supply Case}

In this section we identify an SPM that yields a universal $1/8$-approximation for limited supply auctions. In this case, we have $k$ items to sell, where $k$ can be less than the number of bidders $n$, and this allocation constraint imposes an additional challenge: using the posted price identified in the previous section will cause us to hit the supply constraint without having collected enough revenue. To define the price to use in our posted-price mechanism in this context, we apply a trick introduced in \cite{CHMS10} as follows. Given a mechanism that honors the supply constraint, for a fixed bidder, define the allocation probability $q$ to be the probability that she wins in running this mechanism, where the randomization is over all valuation profiles. As the valuations are identically distributed, $q$ is identical for all bidders. The posted price to use is then $p=F^{-1}(1-q)$. The key for us is then to find the right mechanism to draw the allocation probability from. Recall that the optimal risk-neutral mechanism is the VCG mechanism with reserve $p^*$. In order to have better control over the distribution of the revenue of the mechanism, we derive the allocation probability from the VCG mechanism with a discounted reserve $p^*q^*$. By Lemma \ref{lem:regular}, at least half of the bidders meet the reserve in expectation, and as we will show it follows that the allocation probability $q$ is bounded between $\frac{k}{2n}$ and $\frac{k}{n}$. Moreover, the loss in expected revenue due to this sub-optimal reserve is bounded. We formalize these in the following two claims.

\begin{lem} \label{lem:vcg-discount}
$Rev(VCG_{r=p^* q^*})\geq 0.5\cdot Rev(VCG_{p^{*}})$.
\end{lem}
\begin{proof}

For notational convenience, let $\hat{R}(p)=p(1-F(p))$.
Fix a bidder $i$, fix the bids $\mathbf{b}_{-i}$ of the other bidders, and let $t$ be the
threshold induced by the $VCG$ mechanism (with no reserve) for bidder
$i$. Then the threshold bids of bidder $i$ in $VCG_{p^{*}}$ and
$VCG_{r}$ (with $r=p^*q^*$) are $\max\{t,p^{*}\}$ and $\max\{t,r\}$ respectively.
It suffices to show that the expected revenue of bidder $i$ in $VCG_{r}$,
which is $\hat{R}(\max\{t,r\})$, is at least half of that in $VCG_{p^{*}}$,
which is $\hat{R}(\max\{t,p^{*}\})$, and our claim follows by integrating over all $\mathbf{b}_{-i}$ and $i$.

There are two cases. If $t\geq p^{*}$, then $t\geq p^{*} q^* = r$, and so the offered prices and the expected revenues from the two auctions are identical.
Otherwise, $t < p^{*}$, so bidder $i$ is offered $p^{*}$ (with revenue $p^*q^*$) by $VCG_{p^{*}}$, and a price in the interval $[p^*q^*, p^{*}]$ by $VCG_{r}$.
As revenue is monotonically decreasing as price goes down from $p^*$ to $0$,
the revenue of $VCG_{r}$ is minimized when the offer price is $p^*q^*$. By Lemma~\ref{lem:regular} the resulting revenue $p^*q^* \left(1-F(p^*q^*)\right)$ is at least $\frac{p^*q^*}{2}$; integrating over all $\mathbf{b}_{-i}$ and $i$ completes the proof.

\end{proof}

\begin{lem}
\label{cla:q_hat}Let $q$ be the allocation probability of any fixed bidder. Then $q$ lies in the interval $[\frac{k}{2n}, \frac{k}{n}]$.
\end{lem}
\begin{proof}
Let $X$ be the number of bidders with bids at least
$r$. The expected number of winners of $VCG_{r}$ is $\min(k,X)$.
By definition of $q$, $qn$ is the expected number of winners
in $VCG_{r}$. So, $qn=E[\min(k,X)]$ and hence, $q \leq k/n$.

By definition of $r$, each bidder's
bid is at least $r$ with probability at least $0.5$, and so,
$E[X]\geq 0.5 n$. Therefore $qn=E[\min(k,X)]\geq E[\frac{k}{n}\cdot X]=\frac{k}{n}0.5 n=0.5 k$.
\end{proof}

Now we can define our $Hedge$ mechanism (for the limited-supply case). The hedge mechanism is an SPM which makes a take-it-or-leave-it offer at price $p=F^{-1}(1-q)$ to bidders one by one, as long as the supply lasts.

\begin{thm}
\label{thm:limited-regular-1}In a multi-unit auction with $k$ items
and $n$ bidders, where bidders' valuations are drawn i.i.d.\ from
a regular distribution $F$, the $Hedge$ mechanism is a universal $1/8$-approximation to optimal expected utility. \end{thm}

Notice that the revenue of $Hedge$ is $p\cdot\min(Y,k)$, where
$Y$ is the number of bidders who bid at least $p$, which is a binomial variable with parameter $(n,q)$. Hence
$E[Y]=qn\geq 0.5 k$. Crucial to our analysis is the following property about ``capped'' binomial variables:

\begin{lem}
Let $Y$ be a binomial random variable with parameter $(n,q)$ where
 $qn\geq 0.5k$ for some positive integer $k$, then $E[\min(Y,qn)]\geq0.25\cdot qn$.
\end{lem}

\begin{proof}
Clearly $E[Y]=qn\geq 0.5 k$.

First let $k=1$, and hence $ 0.5 \leq qn\leq1$. Note that $E[\min(Y,qn)]/qn=Pr[Y>0]=1-(1-q)^{n}$,
which is at least $1-(1- 0.5 /n)^{n}\geq 1-e^{-0.5}>0.25$.

Next let $k>1$, and hence $qn\geq 0.5  k\geq 1 $. By a result of \cite{KB-80}, one
of $\lceil qn\rceil,\lfloor qn\rfloor$ is the median of $Y$, and
hence $Pr[Y\geq\lfloor qn\rfloor]\geq0.5$. It follows that $E[\min(Y,qn)]\geq Pr[Y\geq\lfloor qn\rfloor]\cdot\lfloor qn\rfloor\geq0.5\cdot\lfloor qn\rfloor \geq0.25qn$, and our claim follows.
\end{proof}

We now complete the proof of Theorem~\ref{thm:limited-regular-1}.

\begin{proof} (of Theorem~\ref{thm:limited-regular-1})
The expected utility of $Hedge$:
\begin{eqnarray*}
E_{\mathbf{v}}[u(p\cdot\min(Y,k))] &\geq & E_{\mathbf{v}}[u(p\cdot\min(Y,qn))] \\
&\geq& E_{\mathbf{v}}[u(pqn)\cdot\frac{\min(Y,qn)}{qn}]\\
&\geq& 1/4 \cdot u(pqn) \\
&\geq& 1/4 \cdot u(Rev(VCG_{r})) \\
&\geq& 1/8 \cdot u(Rev(VCG_{p^{*}}))\\
&\geq& 1/8 \cdot \mbox{Optimal Expected Utility} \\
\end{eqnarray*}

The second step is by concavity of $u$, the fourth step is by monotonicity of the utility function with the following additional justification. For any bidder $i$, she wins with probability $q$ in $VCG_{r}$. On the other hand, the optimal way to maximize expected revenue
subject to the constraint that she wins with probability $q$
is to set a single price $p$ and get expected revenue $qp$.  The fifth step is by Lemma~\ref{lem:vcg-discount}. Applying Fact~\ref{fact:upperbound} completes the proof.
\end{proof}

We do not have an analog of Theorem~\ref{thm:lb} for the limited supply case. We do not know if our analysis is tight (though we can tweak various parameters to improve the ratio slightly) or if it possible to identify a better posted-price mechanism.

\section{The VCG Mechanism} \label{sec:vcg}

In this section, we quantify the universal approximation ratio of the VCG mechanism in multi-unit auctions. This is useful because the VCG mechanism ($k+1$-st price auction) or a  variation of it with a reserve price is often used in practice.

\subsection{The Single-Item Case}

We first restrict our attention to single-item auctions. The main result of this subsection is that the Vickrey mechanism is a universal $(1-1/n)$-approximation when there are $n$ bidders.

\begin{thm}\label{thm:Vickrey}
For a single item auction with $n$ bidders, when valuations are drawn i.i.d.\ from a regular distribution $F$, the Vickrey mechanism is a universal $(1-1/n)$-approximation to optimal expected utility.
\end{thm}

This theorem is a generalization of a result of Dughmi et al.~\cite{Dughmi09}, which was for the risk-neutral case. Most of the proof steps are similar, and so we only mention the proof structure, which is also used in the next section. Let $OPT'$ be the mechanism which first runs the utility-optimal mechanism $OPT$ on the $n-1$ bidders, and then allocates the item for free to the other bidder in case it is still available. Our theorem follows from three statements. First, the revenue (and hence utility) of $OPT'$ on $n$ bidders is equal to that of $OPT$ on $n-1$ bidders. Second, among all mechanisms that always sell the item, including Vickrey and $OPT'$, Vickrey maximizes the winner's valuation and hence virtual utility, and therefore by the characterization of Lemma~\ref{lem:virtual_utility}, Vickrey on $n$ bidders has a higher expected utility than that of $OPT'$ on $n-1$ bidders. Third, as we will show more more generally in Lemma~\ref{lem:removing-k-bidders}, the optimal expected utility from $n-1$ bidders is at least $1-1/n$ fraction of that from $n$ bidders. These three statements altogether imply our theorem.

\subsection{The Multi-Unit Case}

In this section we prove a result analogous to Theorem~\ref{thm:Vickrey} for multi-unit auctions.

\begin{thm}
\label{thm:vcg}In a multi-unit auction with $k$ items and $n$ bidders,
where bidders' valuations are drawn i.i.d.\ from a regular distribution
$F$, the VCG mechanism is a universal $(n-k)/4n$-approximation to optimal expected utility.
\end{thm}

The result implies that as long as the number of bidders is a small multiple of the number of items, the universal approximation ratio of VCG mechanism is close to $1/4$. The proof structure is similar to that of Theorem~\ref{thm:Vickrey}, but the details are different because  Lemma~\ref{lem:virtual_utility} does not extend to the multi-unit case (as discussed in Section~\ref{sec:benchmark}). Recall that the revenue of the VCG mechanism is exactly $k$ times the $k+1$-st highest bid (let the $n+1$-th highest bid be 0). The following probability bound on the $k+1$-st highest bid is crucial to our analysis.

\begin{lem}
\label{lem:tail}
For any regular distribution $F$, and $1<t\leq n$, let $Y$ be the $t$-th largest of $n$ i.i.d.\ random draws from $F$, then $Pr[Y\geq E[Y]]\geq 1/4$.
\end{lem}

\begin{proof}
Our proof consists of two steps. First, given a regular distribution $F$, we construct a slightly non-regular distribution $\tilde{F}$ such that $Pr[Y\geq E[Y]]\geq Pr[\tilde{Y}\geq E[\tilde{Y}]]$, where $\tilde{Y}$ is the $t$-th largest valuation of $n$ i.i.d.\ draws from $\tilde{F}$.
This new distribution $\tilde{F}$ has corresponding revenue function $\tilde{R}(q)=a\cdot q+b$ for $q\in (0,1]$ for
some $b>0$ and $a+b\geq0$, and it then suffices to show that $Pr[\tilde{Y}\geq E[\tilde{Y}]]\geq 1/4$ for such distributions.

Given any regular distribution $F$, let $z=1-F(E[Y])$, and consider the distribution $\tilde{F}$ corresponding to the revenue
function $\tilde{R}$ such that $\tilde{R}(z)=R(z)$ and $\tilde{R}'(q)=R'(z)$ for all $q\in(0,1]$.
In other words, $\tilde{R}$
is the line segment that is tangent with $R$ at $z$. By concavity
of $R$, we have $\tilde{R}(q)\geq R(q)$ for
all $q\in(0,1]$.

To aid the analysis, let $Q_{t,n}$ be the $t$-th order statistics (i.e., the $t$-th smallest valuation) of $n$ i.i.d.\ draws from the uniform distribution over $[0,1]$. Therefore for all $y$, $Pr[Y\geq y]=Pr[Q_{t,n}\leq 1-F(y)]$ and similarly for $\tilde{Y}$ and $\tilde{F}$.
Let $\tilde{z}=1-\tilde{F}(E[\tilde{Y}])$.
Then to show that $Pr[Y\geq E[Y]]\geq Pr[\tilde{Y}\geq E[\tilde{Y}]]$, it suffices to show that $Pr[Q_{t,n}\leq\tilde{z}]\leq Pr[Q_{t,n}\leq z]$, or simply that $\tilde{z}\leq z$.

Recall that $\tilde{R}(q)\geq R(q)$ for all $q$.
Therefore $\tilde{F}(v)\leq F(v)$ for all $v$, and hence $E[\tilde{Y}]\geq E[Y]$.
Also recall that  $\tilde{R}(z)=R(z)$.
Therefore $\tilde{F}^{-1}(1-z)=F^{-1}(1-z)=E[Y]\leq E[\tilde{Y}]=\tilde{F}^{-1}(1-\tilde{z})$. So $\tilde{z}\leq z$.

Now we prove that $Pr[\tilde{Y}\geq E[\tilde{Y}]]\geq 1/4$. Let distribution $\tilde{F}$ be such that the corresponding revenue function
is $\tilde{R}(q)=a\cdot q+b$ for some $b\geq0$ and $a+b\geq0$. Let $f_{t,n}(q)=\frac{n!}{(t-1)!(n-t)!}q^{t-1}(1-q)^{n-t}$
be the density function of $Q_{t,n}$. Then 

\begin{eqnarray*}
E[\tilde{Y}] & =& \int_{q=0}^{1}f_{t,n}(q)\cdot\frac{\tilde{R}(q)}{q}dq\\
&=&\int_{q=0}^{1}f_{t,n}(q)\cdot(a+\frac{b}{q})dq\\
&=&a+b\cdot\frac{n}{t-1},
\end{eqnarray*}
 where we use the facts that $\frac{1}{q}\cdot f_{t,n}(q)=\frac{n}{t-1}\cdot f_{t-1,n-1}(q)$
and that $f_{t,n}$ and $f_{t-1,n-1}$ as density functions both integrate to $1$.
Note that when $q=\frac{t-1}{n}$, $\tilde{R}(q)/q=a+b/q=E[\tilde{Y}]$. Therefore
$1-\tilde{F}(E[\tilde{Y}])=\frac{t-1}{n}$, and hence $Pr[\tilde{Y}\geq E[\tilde{Y}]]=Pr[Q_{t,n}\leq\frac{t-1}{n}]$.

Note that for $n$ i.i.d.\ draws from the uniform distribution over $[0,1]$, the $t$-th order statistic is at most $\frac{t-1}{n}$ if and only if the number of draws that are at most $\frac{t-1}{n}$ is at least $t$. Let $B$ be this number, which is a binomial variable with parameter $n$ and $\frac{t-1}{n}$. Then $Pr[Q_{t,n}\leq\frac{t-1}{n}]=Pr[B\geq t]$, and by properties of binomial distribution, $Pr[B\geq t]$ is at least
$1/4$, where $1/4$ is achieved when $t=n=2$.

\end{proof}

Based on Lemma \ref{lem:tail}, we can prove the following risk-averse version of the classical result of Bulow and Klemperer \cite{BK96}. 
(We suspect that an exact version holds without the approximation factor $1/4$; to recover the statement original result, replace `utility' by `revenue' and remove the `$1/4$'.)

\begin{lem}
\label{lem:bk-utility}Suppose valuations of bidders are drawn i.i.d.\ from a regular distribution. The optimal expected utility when selling $k$ items to $n$ bidders is at most $1/4$
times the expected utility of the VCG mechanism when selling $k$ items to $n+k$ bidders.
\end{lem}
\begin{proof}

We will let superscripts in $VCG^{k,n}$ or $Mye^{k,n}$ denote that
we are selling $k$ items to $n$ bidders. By Fact \ref{fact:upperbound},
the optimal expected utility of selling $k$ items to $n$ bidders is at most
$u(E_{\mathbf{v}}[Rev(Mye^{k,n},\mathbf{v})])$, which by the classic Bulow-Klemperer result \cite{BK96} and the monotonicity of $u$ is at most $u(E_{\mathbf{v}}[Rev(VCG^{k,n+k},\mathbf{v})])$.
Note that the revenue of $VCG^{k,n+k}$ is $k$ times $Y=F^{-1}(1-Q_{k+1,n+k})$,
where $Q_{k+1,n+k}$ is the $k+1$-th order statistics of $n+k$ i.i.d.\ draws from a uniform distribution over $[0,1]$. By Lemma \ref{lem:tail}, we have $Pr[Y\geq E[Y]]\geq1/4$. Our lemma follows
because the utility of $VCG^{k,n+k}$ would be at least $1/4\cdot u(k\cdot E[Y])=1/4\cdot u(E_{\mathbf{v}}[Rev(VCG^{k,n+k},\mathbf{v})])$.
\end{proof}

The following claim bounds the loss of optimal utility in dropping $k$ bidders.

\begin{lem}
\label{lem:removing-k-bidders}Suppose valuations of bidders are drawn i.i.d.\ from a regular distribution. 
The optimal expected utility when selling $k$ items to $n-k$ bidders is at least $1-k/n$ fraction
of the optimal expected utility when selling $k$ bidders to $n$ bidders.
\end{lem}
\begin{proof}
Let $M$ be a utility-optimal mechanism for selling $k$ items to
$n$ bidders $N=\{1,2,\dots,n\}$. For any subset $S$ of bidders, let random variable
$R_{S}$ be the revenue we collect from $S$ in $M$. Then the expected utility of running $M$ on all bidders is $E_{\mathbf{v}}[u(R_{N})]$.
Suppose we randomly select a set $S$ of size $n-k$. Then we have:
\begin{eqnarray*}
& & E_{\mathbf{v},S}[u(R_{S})]\\
&\geq& E_{\mathbf{v}}[E_S[u(R_{N})\cdot\frac{R_{S}}{R_{N}}]]\\
&=& E_{\mathbf{v}}[u(R_{N})\cdot E_S[\frac{R_{S}}{R_{N}}]]\\
&=&(1-\frac{k}{n})\cdot E_{\mathbf{v}}[u(R_{N})]
\end{eqnarray*}

Here the inequality is by the concavity of $u$ and that $R_{S}\leq R_{N}$,
and the second equality is due to the fact that every bidder's revenue is
accounted in $R_{S}$ with probability $1-k/n$. By an averaging argument,
for some set $S$ of $n-k$ bidders, and for some fixed bids $\mathbf{v}_{-S}$
of bidders outside of $S$, the mechanism $M$ induced on $S$ has
expected utility that is at least $1-k/n$ fraction of the expected utility of running $M$ on all bidders. Our lemma follows
because the utility-optimal mechanism on $n-k$ bidders can only do
better than this induced mechanism.
\end{proof}

Now Theorem \ref{thm:vcg} follows by chaining the inequalities from Lemma \ref{lem:bk-utility} and Claim \ref{lem:removing-k-bidders}.

\section{Conclusions and Open Problems}
\label{sec:open}

In this paper, we identify truthful mechanisms for multi-unit auctions that offer universal constant-factor approximations for all risk-averse sellers, no matter what their levels of risk-aversion are. We hope that this paper spurs interest in the design and analysis of mechanisms for risk-averse sellers.

We see several open directions. For instance, identifying better mechanisms for the auction settings studied in this paper, identifying mechanisms for more combinatorial auction settings, and designing online mechanisms that adapt prices based on previous sales. We conclude by singling out a specific challenge: can we characterize the utility-optimal mechanism for a seller with a fixed known utility function? What if the seller's utility function has additional structure--for instance, it satisfies constant (absolute or relative) risk aversion? (Section~\ref{sec:benchmark} discusses how the standard approach from Myerson~\cite{M81} does not work for multi-item auctions.)

%\bibliographystyle{abbrv}
%\bibliography{ec10}  % sigproc.bib is the name of the Bibliography in this case

\begin{thebibliography}{10}

\bibitem{BK96}
J.~Bulow and P.~Klemperer.
\newblock Auctions versus negotiations.
\newblock {\em American Economic Review}, 86(1):180--194, 1996.

\bibitem{CHMS10}
S.~Chawla, J.~Hartline, D.~Malec, and B.~Sivan.
\newblock Sequential posted pricing and multi-parameter mechanism design.
\newblock In {\em Proc. 41st ACM Symp. on Theory of Computing (STOC)}, 2010.

\bibitem{cla-71}
E.~H. Clarke.
\newblock Multipart pricing of public goods.
\newblock {\em Public Choice}, 11:17--33, 1971.

\bibitem{DRY10}
P.~Dhangwatnotai, T.~Roughgarden, and Q.~Yan.
\newblock Revenue maximization with a single sample.
\newblock In {\em Proc. 12th ACM Conf. on Electronic Commerce (EC)}, 2010.

\bibitem{Dughmi09}
S.~Dughmi, T.~Roughgarden, and M.~Sundararajan.
\newblock Revenue submodularity.
\newblock In {\em EC '09: Proceedings of the tenth ACM conference on Electronic
  commerce}, pages 243--252, New York, NY, USA, 2009. ACM.

\bibitem{risk99}
P.~Eso and G.~Futo.
\newblock Auction design with a risk averse seller.
\newblock {\em Economics Letters}, 65(1):71--74, October 1999.

\bibitem{Ewe09}
C.~Ewerhart.
\newblock Optimal design and $\rho$-concavity.
\newblock Working Paper, 2009.

\bibitem{GH-05}
A.~Goldberg and J.~Hartline.
\newblock Collusion-resistant mechanisms for single-parameter agents.
\newblock In {\em Proc. 16th ACM Symp. on Discrete Algorithms}, 2005.

\bibitem{gro-73}
T.~Groves.
\newblock Incentives in teams.
\newblock {\em Econometrica}, 41:617--631, 1973.

\bibitem{hu10}
A.~Hu, S.~A. Matthews, and L.~Zou.
\newblock Risk aversion and optimal reserve prices in first and second-price
  auctions.
\newblock Working Paper, 2010.

\bibitem{buyers81}
E.~S. Maskin and J.~G. Riley.
\newblock Optimal auctions with risk averse buyers.
\newblock {\em Econometrica}, 52(6):1473--1518, November 1984.

\bibitem{M81}
R.~Myerson.
\newblock Optimal auction design.
\newblock {\em Mathematics of Operations Research}, 6(1):58--73, 1981.

\bibitem{AGTBook}
N.~Nisan, T.~Roughgarden, E.~Tardos, and V.~V. Vazirani.
\newblock {\em Algorithmic Game Theory}.
\newblock Cambridge University Press, New York, NY, USA, 2007.

\bibitem{KB-80}
J.~B. R.~Kaas.
\newblock Mean, median and mode in binomial distributions.
\newblock {\em Statistica Neerlandica}, 34:13--18, 1980.

\bibitem{Stig1}
M.~Rothschild and J.~E. Stiglitz.
\newblock Increasing risk: I. a definition.
\newblock {\em Journal of Economic Theory}, 2(3):225--243, September 1970.

\bibitem{Stig2}
M.~Rothschild and J.~E. Stiglitz.
\newblock Increasing risk ii: Its economic consequences.
\newblock {\em Journal of Economic Theory}, 3(1):66--84, March 1971.

\bibitem{vic-61}
W.~Vickrey.
\newblock Counterspeculation, auctions, and competitive sealed tenders.
\newblock {\em J. of Finance}, 16:8--37, 1961.

\end{thebibliography}

%\vfill\eject

\appendix
\section{Missing Proofs}
\subsection{Proof of Lemma \ref{lem:u-regular}}

\begin{lem}
\label{lem:u-regular}Let $F$ be a regular distribution. For any
 concave utility function $u$, $\phi^u_F(v)$ is nondecreasing.
\end{lem}
\begin{proof}
Since $F$ is regular, $\phi^u_{F}(v)=(v-\frac{1}{h(v)})'=1+\frac{h'(v)}{h^{2}(v)}\geq0$.
Then:
\begin{eqnarray}
\frac{d\phi_{F}^{u}(v)}{dv} & = & (u(v)-\frac{u'(v)}{h(v)})'\\
 & = & u'(v)-\frac{u''(v)h(v)-u'(v)h'(v)}{h^{2}(v)}\\
 & = & u'(v)\cdot(1+\frac{h'(v)}{h^{2}(v)})-\frac{u''(v)}{h(v)}\\
 & = & u'(v)\cdot\phi_{F}^u(v)-\frac{u''(v)}{h(v)}\\
 & \geq & 0\end{eqnarray}
\end{proof}

\end{document}